\documentclass[12pt, reqno]{amsart}
\usepackage{a4wide,palatino,epsfig,latexsym,
amsmath,amsthm,graphicx,hyperref,enumerate,morefloats,color, natbib, soul}

\usepackage{tikz}

\title{Sign epistasis and the geometry of interactions}

        \author{Kristina Crona } 
        
\email{kcrona@american.edu}
        
\theoremstyle{plain}

\newtheorem{theorem}{Theorem}[section]

\newtheorem{lemma}[theorem]{Lemma}

\newtheorem{proposition}[theorem]{Proposition}

\newtheorem{corollary}[theorem]{Corollary}

\theoremstyle{definition}

\newtheorem{definition}[theorem]{Definition}

\newtheorem{remark}[theorem]{Remark}

\newtheorem{example}[theorem]{Example}

\newcommand{\R}{{\mathbb R}}

\begin{document}

\maketitle


\begin{abstract}
Approaches to gene interactions based on sign epistasis have been highly influential in recent time.
Sign epistasis is useful for relating local and global properties of fitness landscapes,
as well as for analyzing evolutionary trajectories and constraints. 
The geometric theory of gene interactions,
on the other hand,  provides complete information on 
interactions in terms of minimal dependence relations.

We propose a new framework that combines
aspects of both approaches.
In particular, we provide efficient tools for identifying sign epistasis
and related order perturbations in large genetic systems,
with applications to the malaria-causing parasite {\emph{Plasmodium vivax}}.
We found that order perturbations beyond sign epistasis
are prevalent in the drug-free environment,
which agrees well with the observation
that reversed evolution back to the ancestral type is difficult.
As a theoretical application, we  investigate how rank orders  of genotypes
with respect to fitness relates to additivity.
\end{abstract}

\section{introduction}
Approaches to gene interactions based on sign epistasis, a concept
introduced in \citet{wwc}, have been highly influential in recent years.
Sign epistasis is useful for relating local and global properties of fitness landscapes  \citep{wwc, ptk, cgb},
as well as for analyzing evolutionary predictability and other aspects of
evolutionary potential and constraints \citep{wdd, fkd, dk}, 
A system has sign epistasis if the sign of the effect of a mutation,
whether positive or negative, depends on genetic background.
For instance, a biallelic  2-locus system has sign epistasis if 
\[ 
w_{00} >w_{10},   \quad w_{11} >w_{01},
\]
where $w_g$ denotes fitness of the genotype.
Here the mutation $0 \mapsto 1$ at the first locus decreases fitness if the background is $0$,
and increases fitness if the background is 1.

The geometric approach to gene interactions developed in \citet{bps}
provides complete information on interactions in terms of a class of linear forms
(circuits, as defined in matroid theory, see the Result section for details).
For two-locus system, one considers the linear form
\[
u= w_{11}+w_{00}-w_{10}-w_{01}.
\]
Notice  that the conditions $w_{00} >w_{10}$ and $w_{11} >w_{01}$ 
implies that $u>0$. This simple observation suggests 
a relation between approaches focused on sign epistasis and  the geometric approach.

Here,  we investigate how the different approaches relate
and propose a framework that combines aspects of both theories.
In brief, the sign of the effect of
a double mutant, or any higher order mutant, may depend 
on background. Such order perturbations are similar to sign epistasis, except that we
consider the effect of replacing blocks rather than a single locus.
In particular, it is of interest to identify block replacements that are universally
beneficial. Importantly, we need not assume block independence, but rather that the sign of the effect
is independent.

\section{results}
We consider  biallelic $n$-locus systems.
For simplicity, we assume that there exists a total order of the genotypes
with respect to fitness, usually referred to as the rank order.

It is straight forward to analyze sign epistasis in the two-locus case.
A system has sign epistasis if at least one of the following four conditions are satisfied:
\begin{align*}
& w_{00} >w_{10}  \text{ and } w_{11} >w_{01}, \\
& w_{00} > w_{01} \text{ and }  w_{11}  >  w_{10}, \\
& w_{00} <w_{10},   \text{ and }  w_{11} <w_{01}, \\
&w_{00} <w_{01},    \text{ and }  w_{11} <w_{10}. 
\end{align*}

Differently expressed,  a system has
sign epistasis exactly if 
at least one of the following expression is negative.
\begin{align}
&(w_{00}-w_{10})(w_{01}-w_{11})  \\
&(w_{00}-w_{01})(w_{10}-w_{11})  
\end{align}

For the rank order $w_{11}> w_{10} >w_{01} >w_{00}$ both (1) and (2) are positive.
On the other hand, for the rank order
$
w_{11}> w_{00} >w_{10} >w_{01}
$
both (1) and (2) are negative. Such a "double flip" is sometimes referred to as reciprocal sign epistasis \citep{ptk}.

For a general analysis of $n$-locus system, one can consider a set of expression analogous to (1) and (2).
It is convenient to introduce circuits for a systematic description. By using linear algebra, we can
keep the discussion elementary.

For the genotypes $00, 10, 01, 11$,  one can
form vectors in $\R^3$ be adding an extra coordinate 1
\[
(0,0,1), (1,0,1), (0,1,1), (1,1,1) .
\]
The four vectors are linearly dependent,
\[
(1,1,1)+ (0,0,1)-(1,0,1)-(0,1,1)= {\bf{0}}
\]
Notice that the dependence relation corresponds exactly to 
\[
u= w_{11}+w_{00}-w_{10}-w_{01}.
\]
The form $u$ is referred to as a circuit.
In general, circuits are defined as minimal dependence relations,
in the sense that each proper subset of the vectors 
[with non-zero coefficients] are linearly independent.
 (Here each proper subset of the vectors  $(0,0,1)$, $(1,0,1)$, $(0,1,1)$ and $(1,1,1)$
is linearly independent.)
The description works in general, since
one can form vectors in $\R^{n+1}$ from the vertices of the $n$-cube by adding a coordinate 1.
The circuits are the minimal dependence relations, exactly as in the two-locus case.


\begin{remark}
The relation between circuits and sign epistasis can be summarized as follows for $n=2$.
\begin{itemize}
\item[(i)]
A rank order of the genotypes $00$, $10$, $01$ and $11$ with respect to fitness implies that  $u>0$ or $u<0$
exactly if the system has sign epistasis.
\item[(ii)]
From the information that the rank order implies  $u>0$ (or $u<0$) alone, one
cannot determine if there are one or two order perturbations, i.e., whether or not the system
has reciprocal sign epistasis.
\item[(ii)]
The system has reciprocal sign epistasis if both expressions 
(1) and (2) are negative, and sign epistasis if at least one
of them is negative.
\end{itemize}
\end{remark}

A similar approach works for a  general $n$-locus systems, 
although one needs to consider more general order perturbations.
For $n=3$, the set of circuits includes 
\[
w_{111}+w_{000}-w_{100}-w_{011}.
\] 
Similar to the study of the two-locus case, it is natural to consider the two related expressions

\begin{align}
( w_{111}-w_{011}) (w_{100}-w_{000} )  \\
( w_{111}-w_{100}) (w_{011}-w_{000} ).
\end{align}

If the rank order implies that (3) is negative, then the system has
sign epistasis. If the rank order implies that (4) is negative,
then the order perturbation is of a different type.
Specifically, the sign of the effect of a double mutation $00 \mapsto 11$
at the second and third loci depends on background.

\begin{definition}
A system has a rectangular perturbation if the
sign of the effect of replacing a subset of loci, according to the rule $0 \mapsto 1$ and  $1\mapsto 0$,
depends on background. The size of the perturbation refers to the number of loci replaced.
\end{definition}

In particular, a rectangular perturbation of size one is
a case of sign epistasis.
However, the next example shows that a system can have rectangular perturbations
even if it does not have sign epistasis. 

\begin{example}
Consider the system:
\[
w_{000}=1,
w_{100}=1.1,  w_{010}=1.12 , w_{001}=1.09,
w_{110}= 1.2, w_{101}=1.22, w_{011}=1.19,
w_{111}=1.3.
\]
The change $0 \mapsto 1$ at any locus increases fitness
regardless of background. Consequently, there is no sign epistasis.
However, the sign of the effect of the change  $10 \mapsto 01$ 
at the first pair of loci depends on background since
\[
w_{100} <w_{010}, \quad w_{101} > w_{110}.
\]
\end{example}
Example 1 is interesting for another reason.
The example  shows that the rank order 
\[
w_{111}>w_{101} >w_{110} >w_{011} >w_{010} >w_{100} >w_{001}> w_{000},
\]
is incompatible with additive fitness, even though the system has no sign epistasis.

For investigating
rectangular order perturbations,
one can use the class of circuits $\mathcal C$ where exactly four variables have non-zero coefficients
For $n=3$, there are 12 such circuits.
It follows that one can identify all rectangular perturbations by
checking  24 expressions, including (3) and (4) above, (see the appendix for the complete list).

\begin{proposition}
Let $\mathcal C$ be the class of circuits with non-zero coefficients for exactly four elements.
\begin{itemize}
\item[(i)]
Each case of sign epistasis corresponds a signed circuit interaction for an element in  $\mathcal C$. i.e., 
the rank order of genotypes with respect to fitness implies that the circuit is positive or negative.
\item [(ii)]
More generally, each rectangular  perturbation corresponds to a signed circuit interaction in  $\mathcal C$.
\item[(iii)]
Each circuit in $\mathcal C$ corresponds to exactly two (potential) rectangular perturbations.
\end{itemize}
\end{proposition}

\begin{theorem}
The total number of (potential) rectangular  perturbations for an
$n$-locus system is
\[
2  \left( \frac{6^n}{8} - 4^{n-1} + 2^{n-3}  \right) .
\]
Moreover, the number of rectangular perturbations of size $k$ (exactly $k$ loci are  replaced) 
equals
\[
2^{k-1} {n \choose k} {2^{n-k}  \choose 2}.
\]
\end{theorem}
For a proof of this result, see Methods.
In particular, it follows from Theorem 2.5 that the number of potential order perturbations
for a three-locus system is 24, as mentioned (see the appendix).
However, it is not possible that all expressions are negative
for one and the same system.

The next result follows immediately from the theorem.
\begin{corollary}
A complete investigation of sign epistasis for an $n$-locus system, 
requires that one checks the sign of
\[
 n  \, {2^{n-1}  \choose 2},
 \]
expressions.
\end{corollary}

We applied rectangular perturbations to a study of
the malaria-causing parasite {\emph{Plasmodium vivax}} \citep{oh}.
The original study concerns a 4-locus system
exposed to different concentrations of the
anti-malarial drug pyrimethamine (PYR).
The quadruple mutant denoted 1111 has the highest degree 
of drug resistance, whereas the  genotype 0000
has the highest fitness among all genotypes
 in the drug-free environment.

We compared the highest concentration of the drug
and the drug-free environment. 
Sign epistasis were prevalent in both fitness landscapes;
the landscapes had 54 and 55 rectangular 
perturbations of size one, respectively.
However, the latter landscape had
approximately twice as many perturbations of size two and three.
This finding agrees well with the
authors' observation that resistance development
is a relatively straight forward process, whereas
reversed evolution from the four-tuple mutant  $1111$ back  to 
the ancestral type is difficult.

\begin{table}
\begin{tabular}{ l  c   c   c   c}
&  & & & \\
 \hline
Perturbation size   &  1 &  2 &  3 &  size (1-3) \\
\hline 
 Drug-exposed   & 54 &  39 & 9  & 102\\ 
Drug-free & 55 &  21 & 5   &  81\\
\hline
\end{tabular}
\caption{Rectangular perturbations for drug exposed and drug free malaria
 fitness landscapes.
The prevalence of sign epistasis (size 1 perturbations) is similar.
However, rectangular perturbation of size 2 and 3 differ substantially
between the landscapes.}
\end{table}


One can also apply rectangular perturbations for an analysis of rank orders and additivity.
For a three-locus systems, there are $8!=40,320$
rank orders. 
As remarked in \citet{cgg}, exactly 384 orders are compatible
with absence of sign epistasis.
By using the list  in the appendix,
one can verify that 288  orders imply
rectangular order perturbations.
Moreover, one can verify that the remaining 96 orders
are compatible with additive fitness.
In summary, $0.24$ percent of all rank orders for three-locus
systems are compatible with additive fitness.
In fact, after relabeling only the following two orders are
compatible with additive fitness.
\begin{align*}
w_{111}>w_{110}>w_{101}>w_{011}>w_{100}>w_{010}>w_{001}>w_{000} \\
w_{111}> w_{110}>w_{101}>w_{100}>w_{011}>w_{010}>w_{001}>w_{000}.   
\end{align*}

\section{methods}
For counting rectangular perturbations, one needs to find all
rectangles with vertices in an $n$-cube.
We provide a proof for the reader's convenience.
The proof depends on Stirling numbers of the second kind.
We refer to \citet{g} for concepts used in this section.

\begin{lemma}
For Stirling numbers of the second kind $S(n,k)$ the following identities hold.
\[
\begin{aligned}
&S(n,2) =2^{n-1} -1 \\
&S(n,3)=\frac{1}{6}(3^n - 3 \cdot 2^{n} +3)
\end{aligned}
\]
\end{lemma}
\begin{proof}
The first formula holds since there
are $2^n-2$ non-empty proper subsets of $n$ elements,
and each partition corresponds to exactly two subsets.
Similarly, the second formula can be derived 
from the observation that one can construct three labeled subsets
of $n$ elements in $3^n$ ways. After reducing 
for all cases with empty sets, the number of alternatives is
\[
3^n-3(2^n-2)-3=3^n-3 \cdot 2^n+3.
\]
Each partition corresponds to 6 alternatives, which completes the argument.
\end{proof}

\begin{lemma}
There are 
\[
6^n/8 - 4^{n-1} + 2^{n-3} 
\]
rectangles with vertices in an $n$-cube.
\end{lemma}

\begin{proof}
For each vertex $s_1 \dots s_n$ in an $n$-cube, one
can construct a rectangle with vertices
on the $n$-cube as follows.
Distribute the set of $n$ loci into three subsets
$S_1, S_2$ and $S_3$, where the intersection
of each pair of sets is empty, and where
$S_1$ and $S_2$ are non-empty.

From the vertex $s_1 \dots s_n$,
one constructs the remaining vertices
by replacing sets of loci, according to the
rule $0 \mapsto 1$, and $1 \mapsto 0$.
Specifically, one vertex is obtained by replacing all elements in $S_1$,
one  by replacing all elements in $S_2$, 
and the last by replacing all elements in $S_1 \bigcup S_2$.
In case $S_3$ is empty,  one can construct $S(n,2)$ rectangles.
If  all three sets are non-empty, then one can choose 
$S_3$ in three ways, and consequently construct $3 \cdot S(n,3)$ rectangles.
 In total one obtains $3 \cdot S(n,3)+S(n,2)$ rectangles starting from a particular vertex.
There are $2^n$ vertices in the $n$-cube and each rectangle has four vertices.
By the previous lemma, the number of rectangles is
\[\frac{2^n}{4} \cdot ( 3 \cdot S(n,3)+S(n,2) )=6^n/8 - 4^{n-1} + 2^{n-3}, \]
which completes the proof.
\end{proof}


\noindent
We can now prove the main result.

\noindent
{\emph{Proof of Theorem 2.5.}}
The total number of (potential) rectangular  perturbations for an
$n$-locus system is
\[
2  \left( \frac{6^n}{8} - 4^{n-1} + 2^{n-3}  \right).
\]
A rectangle with vertices in the $n$-cube corresponds to exactly two rectangular order perturbations,
one for each pair of parallel edges. Consequently, the result follows from Lemma 3.1.

The second part of the theorem states
that the number of rectangular perturbations where exactly $k$ loci are  replaced  equals
\[
2^{k-1} {n \choose k} {2^{n-k}  \choose 2}.
\]
The positions of the $k$ loci that change can be chosen in ${n \choose k}$ different ways.
There are $2^k$ words of length $k$. One can choose  the word and its replacement
in $2^{k-1}$ different ways (for instance the replacement of  $110$ is $001$).
 Finally, that are $2^{n-k} $ different backgrounds, so 
 that a pair of backgrounds can be chosen  in  $ {2^{n-k}  \choose 2}$ ways.


\newpage

\section{Appendix}

{\emph{The complete list of expressions for identifying rectangular perturbations for  $n=3$}}.
The first 18 expressions can identify sign epistasis and the last 6 expressions
rectangular perturbations where two loci are replaced.

\medskip
\[
\begin{aligned}
(w_{000}-w_{100})(w_{010}-w_{110}) \\
(w_{000}-w_{100})(w_{001}-w_{101}) \\
(w_{000}-w_{100})(w_{011}-w_{111})  \\
(w_{010}-w_{110})(w_{001}-w_{101}) \\
(w_{010}-w_{110})(w_{011}-w_{111}) \\
(w_{001}-w_{101})(w_{011}-w_{111})\\
(w_{000}-w_{010}) (w_{100}-w_{110})\\
(w_{000}-w_{010}) (w_{001}-w_{011})\\
(w_{000}-w_{010}) (w_{101}-w_{111})\\
(w_{100}-w_{110}) ( w_{001}-w_{011})\\
(w_{100}-w_{110})  (w_{101}-w_{111})\\
(w_{001}-w_{011}) (w_{101}-w_{111})\\
(w_{000}-w_{001}) (w_{100}-w_{101})\\
(w_{000}-w_{001}) (w_{010}-w_{011})\\
(w_{000}-w_{001}) (w_{110}-w_{111})\\
(w_{100}-w_{101})  (w_{010}-w_{011})\\
(w_{100}-w_{101}) (w_{110}-w_{111})\\
(w_{010}-w_{011}) (w_{110}-w_{111}) \\
(w_{000}-w_{110} ) (w_{001}-w_{111})\\
(w_{000}-w_{101})(w_{010}-w_{111})\\
(w_{000}-w_{011}) (w_{100}-w_{111})\\
(w_{100}-w_{010}) (w_{101}-w_{011})\\
(w_{100}-w_{001}) (w_{110}-w_{011})\\
(w_{010}-w_{001}) (w_{110}-w_{101})
\end{aligned}
\]

\bigskip


\end{document}